\RequirePackage{fix-cm}
\documentclass[smallextended]{svjour3}       
\smartqed  
\usepackage{graphicx}
\usepackage{algorithm}
\usepackage{algorithmic}
\usepackage{comment}
\usepackage{amsmath}

\begin{document}

\title{Strongly Stable and Maximum Weakly Stable Noncrossing Matchings
\thanks{A preliminary version of this paper appeared in the proceedings of the 31st International Workshop on Combinatorial Algorithms (IWOCA 2020) \cite{hmo20}.
This work was supported by JSPS KAKENHI Grant Numbers JP16K00017, JP19K12820, and JP20K11677.}
}


\author{Koki Hamada  \and Shuichi Miyazaki  \and Kazuya Okamoto}


\institute{K. Hamada \at
              NTT Corporation, 3-9-11, Midori-cho, Musashino-shi, Tokyo 180-8585, Japan \\
              Graduate School of Informatics, Kyoto University, Yoshida-Honmachi, Sakyo-ku Kyoto 606-8501, Japan \\
              \email{koki.hamada.rb@hco.ntt.co.jp}           
           \and
          S. Miyazaki \at
          Academic Center for Computing and Media Studies, Kyoto University, Yoshida-Honmachi, Sakyo-ku, Kyoto 606-8501, Japan \\
              \email{shuichi@media.kyoto-u.ac.jp}           
           \and
          K. Okamoto \at
          Division of Medical Information Technology and Administration Planning, Kyoto University Hospital, 54 Kawaharacho, Shogoin, Sakyo-ku, Kyoto 606-8507, Japan\\
              \email{kazuya@kuhp.kyoto-u.ac.jp}           
}

\date{}

\sloppy

\maketitle

\begin{abstract}
In IWOCA 2019, Ruangwises and Itoh introduced {\em stable noncrossing matchings}, where participants of each side are aligned on each of two parallel lines, and no two matching edges are allowed to cross each other.  
They defined two stability notions, {\em strongly stable noncrossing matching} ({\em SSNM}) and {\em weakly stable noncrossing matching} ({\em WSNM}), depending on the strength of blocking pairs.
They proved that a WSNM always exists and presented an $O(n^{2})$-time algorithm to find one for an instance with $n$ men and $n$ women. 
They also posed open questions of the complexities of determining existence of an SSNM and finding a largest WSNM.  
In this paper, we show that both problems are solvable in polynomial time.
Our algorithms are applicable to extensions where preference lists may include ties, except for one case which we show to be NP-complete.
This NP-completeness holds even if each person's preference list is of length at most two and ties appear in only men's preference lists.
To complement this intractability, we show that the problem is solvable in polynomial time if the length of preference lists of one side is bounded by one (but that of the other side is unbounded).
\keywords{Stable marriage \and Noncrossing matching \and Polynomial-time algorithms \and NP-completeness.}
\end{abstract}

\section{Introduction}\label{sec:intro}
In the classical {\em stable marriage problem} \cite{gs62}, there are two sets of participants, traditionally illustrated as men and women, where each person has a {\em preference list} that orders a subset of the members of the opposite gender.
This variant is called the stable marriage with incomplete liests, or {\em SMI} for short.
A {\em matching} is a set of (man, woman)-pairs where no person appears more than once.
A {\em blocking pair} for a matching $M$ is (informally) a pair of a man and a woman who are not matched together in $M$ but both of them become better off if they are matched.
A matching that admits no blocking pair is a {\em stable matching}.
The stable marriage problem is one of the recently best-studied topics, with a lot of applications to matching and assignment systems, such as high-school match \cite{apr05,aprs05} and medical resident assignment \cite{roth84}.
See some textbooks \cite{knuth76,gi89,rs90,man13} for more information.

Recently, Ruangwises and Itoh \cite{ri19} incorporated the notion of noncrossing matchings \cite{atal85,clw15,kt86,mop93,ww85} to the stable marriage problem.
In their model, there are two parallel lines where $n$ men are aligned on one line and $n$ women are aligned on the other line.
A matching is {\em noncrossing} if no two edges of it cross each other.  
A {\em stable noncrossing matching} is a matching which is simultaneously stable and noncrossing.
They defined two notions of stability: 
In a {\em strongly stable noncrossing matching} ({\em SSNM}), the definition of a blocking pair is the same as that of the standard stable marriage problem.
Thus the set of SSNMs is exactly the intersection of the set of stable matchings and that of noncrossing matchings.
In a {\em weakly stable noncrossing matching} ({\em WSNM}), a blocking pair has an additional condition that it must not cross matching edges.
Ruangwises and Itoh \cite{ri19} proved that a WSNM exists for any instance, and presented an $O(n^{2})$-time algorithm for the problem of find a WSNM (denoted {\sc Find\_WSNM}).
They also demonstrated that an SSNM does not always exist, and that there can be WSNMs of different sizes.
Concerning these observations, they posed open questions on the complexities of the problems of determining the existence of an SSNM (denoted {\sc Exist\_SSNM}) and finding a WSNM of maximum cardinality (denoted {\sc Max\_WSNM}).

\bigskip

\noindent
{\bf Our Contributions.} 
Table \ref{tbl:results} summarizes previous and our results, where our results are described in bold.
We first show that both the above mentioned open problems are solvable in polynomial time.
Specifically, {\sc Exist\_SSNM} is solved in $O(n^{2})$-time by exploiting the well-known Rural Hospitals theorem (Proposition~\ref{prop:SSNM}) and {\sc Max\_WSNM} is solved in $O(n^{4})$-time by an algorithm based on dynamic programming (Theorem~\ref{thm:MAX-WSNM}).

We then consider extended problems where preference lists may include ties.
When ties are allowed in preference lists, the problem is denoted {\em SMTI} and there are three stability notions, {\em super-}, {\em strong}, and {\em weak} stabilities \cite{ir94}.
We show that our algorithm for solving {\sc Max\_WSNM} is applicable to all of the three stability notions with slight modifications (Corollary \ref{coro:MAX-WSNM-super-strong-weak}).
We also show that our algorithm for solving {\sc Exist\_SSNM} can be applied to super- and strong stabilities without any modification (Corollaries \ref{coro:SSNM-super} and \ref{coro:SSNM-strong}).
In contrast, we show that {\sc Exist\_SSNM} is NP-complete for the weak stability (Theorem~\ref{thm:SSNM-weak}).

This NP-completeness holds even for a restricted case where the length of each person's preference list is at most two and ties appear in only men's preference lists.
To complement this intractability, we show that if each man's preference list contains at most one woman (but women's preference lists may be of unbounded length), the problem is solvable in $O(n)$-time (Theorem \ref{thm:SSNM-weak-positive}).
If we parameterize this problem by two positive integers $p$ and $q$ that bound the lengths of preference lists of men and women, respectively, Theorems \ref{thm:SSNM-weak} shows that the problem is NP-complete even if $p \leq 2$ and $q \leq 2$, while Theorem \ref{thm:SSNM-weak-positive} shows that the problem is solvable in polynomial time if $p=1$ or $q=1$ (by symmetry of men and women).  
Thus the computational complexity of the problem is completely solved in terms of the length of preference lists.
We remark that this is a rare case since many NP-hard variants of the stable marriage problem can be solved in polynomial time if the length of preference lists of one side is bounded by two \cite{imm09,bmm10,bmm12,mo19}.

\begin{table}[htb]
\begin{center}
\renewcommand{\arraystretch}{1.3}
\caption{Previous and our results (our results in bold).}\label{tbl:results}
\smallskip
  \begin{tabular}{ll|l|l|l} \hline
     & & \ {\sc Exist\_SSNM} & \ {\sc Find\_WSNM} & \ {\sc Max\_WSNM} \\ \hline 
   SMI & & \ {\bf\boldmath $O(n^{2})$ [Proposition~\ref{prop:SSNM}]} \ & \ $O(n^{2})$ \cite{ri19} \ & \ {\bf\boldmath $O(n^{4})$ [Theorem~\ref{thm:MAX-WSNM}]} \ \\ \hline 
    SMTI  & super & \ {\bf\boldmath $O(n^{2})$ [Corollary~\ref{coro:SSNM-super}]} & & \ {\bf\boldmath $O(n^{4})$ [Corollary~\ref{coro:MAX-WSNM-super-strong-weak}]}  \\ \cline{2-5}
   & strong \ & \ {\bf\boldmath $O(n^{3})$ [Corollary~\ref{coro:SSNM-strong}]} & & \ {\bf\boldmath $O(n^{4})$ [Corollary~\ref{coro:MAX-WSNM-super-strong-weak}]}  \\ \cline{2-5}
      & weak & \ {\bf NPC$^{*1}$ [Theorem~\ref{thm:SSNM-weak}]} & & \ {\bf\boldmath $O(n^{4})$ [Corollary~\ref{coro:MAX-WSNM-super-strong-weak}]}  \\
      & & \ {\bf\boldmath $O(n)$}$^{*2}$ {\bf [Theorem~\ref{thm:SSNM-weak-positive}]} & &   \\ \hline
  \end{tabular}
\end{center}
$^{*1}$ even if each person's preference list contains at most two persons and ties appear in only men's preference lists.\\
$^{*2}$ if each man's preference list contains at most one woman.
\end{table}

\bigskip

\noindent
{\bf Progress from the Conference Version.}  A preliminary version of this paper appeared in the proceedings of the 31st International Workshop on Combinatorial Algorithms (IWOCA 2020) \cite{hmo20}.
In \cite{hmo20}, only NP-completeness was shown for {\sc Exist\_SSNM} in the weak stability in SMTI (of unbounded-length preference lists).
In the current manuscript, we investigated the computational complexity of this problem in terms of the length of preference lists:
We strengthened the reduction in the proof of Theorem \ref{thm:SSNM-weak} to show that NP-completeness holds even if the length of preference lists is at most two.
Moreover, we added Theorem \ref{thm:SSNM-weak-positive} that shows that the problem can be solved in polynomial time if the length of preference lists of one-side is at most one.
As mentioned previously, these two theorems solve the computational complexity of this problem in terms of the length of preference lists.

\section{Preliminaries}\label{sec:pre}

In this section, we give necessary definitions and notations, some of which are taken from Ruangwises and Itoh \cite{ri19}.
An instance consists of $n$ men $m_{1}, m_{2}, \ldots, m_{n}$ and $n$ women $w_{1}, w_{2}, \ldots, w_{n}$.
We assume that the men are lying on a vertical line in an increasing order of indices from top to bottom, and similarly the women are lying in the same manner on another vertical line parallel to the first one.
Each person has a preference list over a subset of the members of the opposite gender.
For now, assume that preference lists are strict, i.e., do not contain ties.
We call such an instance an {\em SMI-instance}.
If a person $q$ appears in a person $p$'s preference list, we say that $q$ is {\em acceptable} to $p$.
If $p$ and $q$ are acceptable to each other, we say that $(p,q)$ is an {\em acceptable pair}.
We assume without loss of generality that acceptability is mutual, i.e., $p$ is acceptable to $q$ if and only if $q$ is acceptable to $p$.
If $p$ prefers $q_{1}$ to $q_{2}$, then we write $q_{1} \succ_{p} q_{2}$.

A {\em matching} is a set of acceptable pairs of a man and a woman in which each person appears at most once.
If $(m,w) \in M$, we write $M(m)=w$ and $M(w)=m$.
If a person $p$ is not included in a matching $M$, we say that $p$ is {\em single} in $M$ and write $M(p)=\emptyset$.
Every person prefers to be matched with an acceptable person rather than to be single, i.e., $q \succ_{p} \emptyset$ holds for any $p$ and any $q$ acceptable to $p$.

A pair in a matching can be seen as an edge on the plane, so we may use ``pair'' and ``edge'' interchangeably.
Two edges $(m_{i}, w_{j})$ and $(m_{x}, w_{y})$ are said to {\em cross} each other if they share an interior point, or formally, if $(x-i)(y-j) < 0$ holds.
A matching is {\em noncrossing} if it contains no pair of crossing edges.

For a matching $M$, an acceptable pair $(m, w) \not\in M$ is called a {\em blocking pair} for $M$ (or $(m, w)$ {\em blocks} $M$) if both $w \succ_{m} M(m)$ and $m \succ_{w} M(w)$ hold.
A {\em noncrossing blocking pair} for $M$ is a blocking pair for $M$ that does not cross any edge of $M$.
A matching $M$ is a {\em weakly stable noncrossing matching} ({\em WSNM}) if $M$ is noncrossing and does not admit any noncrossing blocking pair. 
A matching $M$ is a {\em strongly stable noncrossing matching} ({\em SSNM}) if $M$ is noncrossing and does not admit any blocking pair.
Note that an SSNM is always a WSNM by definition but the converse is not true.


We then extend the above definitions to the case where preference lists may contain ties.  
A {\em tie} of a person $p$'s preference list is a set of one or more persons who are equally preferred by $p$, and $p$'s preference list is a strict order of ties.  
We call such an instance an {\em SMTI-instance}.  
In a person $p$'s preference list, suppose that a person $q_{1}$ is in tie $T_{1}$, $q_{2}$ is in tie $T_{2}$, and $p$ prefers $T_{1}$ to $T_{2}$.  
Then we say that $p$ {\em strictly prefers} $q_{1}$ to $q_{2}$ and write $q_{1} \succ_{p} q_{2}$.  If $q_{1}$ and $q_{2}$ are in the same tie (including the case that $q_{1}$ and $q_{2}$ are the same person), we write $q_{1} =_{p} q_{2}$.
If $q_{1} \succ_{p} q_{2}$ or $q_{1} =_{p} q_{2}$ holds, we write $q_{1} \succeq_{p} q_{2}$ and say that $p$ {\em weakly prefers} $q_{1}$ to $q_{2}$.

When ties are present, there are three possible definitions of blocking pairs, and accordingly, there are three stability notions, {\em super-stability}, {\em strong stability}, and {\em weak stability} \cite{ir94}:

\begin{itemize}

\item In the super-stability, a blocking pair for a matching $M$ is an acceptable pair $(m, w) \not\in M$ such that $w \succeq_{m} M(m)$ and $m \succeq_{w} M(w)$.

\item In the strong stability, a blocking pair for a matching $M$ is an acceptable pair $(p, q) \not\in M$ such that $q \succeq_{p} M(p)$ and $p \succ_{q} M(q)$.
Note that the person $q$, who strictly prefers the counterpart $p$ of the blocking pair, may be either a man or a woman. 

\item In the weak stability, a blocking pair for a matching $M$ is an acceptable pair $(m, w) \not\in M$ such that $w \succ_{m} M(m)$ and $m \succ_{w} M(w)$.

\end{itemize}

With these definitions of blocking pairs, the terms ``noncrossing blocking pair'', ``WSNM'', and ``SSNM'' for each stability notion can be defined analogously.
In the SMTI case, we extend the names of stable noncrossing matchings using the type of stability as a prefix.
For example, a WSNM in the super-stability is denoted {\em super-WSNM}.

Note that, in this paper, the terms ``weak'' and ``strong'' are used in two different meanings.
This might be confusing but we decided not to change these terms, respecting previous literature.

For implementation of our algorithms, we use ranking arrays described in Sect.~1.2.3 of \cite{gi89}. 
Although in \cite{gi89} ranking arrays are defined for complete preference lists without ties, they can easily be modified for incomplete lists and/or with ties.
Then, by the aid of ranking arrays, we can determine, given persons $p$, $q_{1}$, and $q_{2}$, whether $q_{1} \succ_{p} q_{2}$ or $q_{2} \succ_{p} q_{1}$ or $q_{1} =_{p} q_{2}$ in constant time. 
Also we can determine, given $m$ and $w$, if $(m,w)$ is an acceptable pair or not in constant time.

\section{Strongly Stable Noncrossing Matchings}\label{sec:SSNM}

\subsection{SMI}\label{sec:SSNM-SMI}

In SMI, an easy observation shows that existence of an SSNM can be determined in $O(n^{2})$ time:

\begin{proposition}\label{prop:SSNM}
There exists an $O(n^{2})$-time algorithm to find an SSNM or to report that none exists, given an SMI-instance.
\end{proposition}

\begin{proof}
Note that an SSNM is a stable matching in the original sense.
In SMI, there always exists at least one stable matching \cite{gi89}, and due to the Rural Hospitals theorem \cite{gs85,roth84,roth86}, the set of matched agents is the same in any stable matching.
These agents can be determined in $O(n^{2})$ time by using the Gale-Shapley algorithm \cite{gs62}.
There is only one way of matching them in a noncrossing manner.
Hence the matching constructed in this way is the unique candidate for an SSNM.
All we have to do is to check if it stable, which can be done in $O(n^{2})$ time. \qed
\end{proof}

\subsection{SMTI}\label{sec:SSNM-SMTI}

In the presence of ties, super-stable and strongly stable matchings do not always exist.
However, there is an $O(n^{2})$-time ($O(n^{3})$-time, respectively) algorithm that finds a super-stable (strongly stable, respectively) matching or reports that none exists \cite{ir94,kmmp04}.
Also, the Rural Hospitals theorem takes over to the super-stability \cite{ims00} and strong stability \cite{ims03}.
Therefore, the same algorithm as in Sect.~\ref{sec:SSNM-SMI} applies for these cases, implying the following corollaries:

\begin{corollary}\label{coro:SSNM-super}
There exists an $O(n^{2})$-time algorithm to find a super-SSNM or to report that none exists, given an SMTI-instance.
\end{corollary}

\begin{corollary}\label{coro:SSNM-strong}
There exists an $O(n^{3})$-time algorithm to find a strong-SSNM or to report that none exists, given an SMTI-instance.
\end{corollary}

In contrast, the problem becomes NP-complete for the weak stability even for a highly restricted case:

\begin{theorem}\label{thm:SSNM-weak}
The problem of determining if a weak-SSNM exists, given an SMTI-instance, is NP-complete, even if each person's preference list contains at most two persons and ties appear in only men's preference lists.
\end{theorem}

\begin{proof}
Membership in NP is obvious.
We show NP-hardness by a reduction from 3SAT \cite{cook71}.  
Its instance consists of a set of variables and a set of clauses.
Each variable takes either true (1) or false (0).  A {\em literal} is a variable or its negation.
A {\em clause} is a disjunction of at most three literals.
A clause is {\em satisfied} if at least one of its literals takes the value 1, and is {\em unsatisfied} otherwise.
A 0/1 assignment to variables that satisfies all the clauses is called a {\em satisfying assignment}.
An instance $f$ of 3SAT is {\em satisfiable} if it has at least one satisfying assignment; otherwise $f$ is {\em unsatisfiable}.
3SAT asks if there exists a satisfying assignment.
3SAT is NP-complete even if each variable appears at most three times, at most twice positively and at most twice negatively, and each clause contains two or three literals \cite{tov84}.
We use 3SAT instances restricted in this way.

Now we show the reduction.
Let $f$ be an instance of 3SAT having $n$ variables $x_{i} (1 \leq i \leq n)$ and $m$ clauses $C_{j} (1 \leq j \leq m)$.
For $k=2, 3$, we call a clause containing $k$ literals a {\em $k$-clause}.
Suppose that there are $m_{2}$ 2-clauses and $m_{3}$ 3-clauses (thus $m_{2}+m_{3}=m$), and assume without loss of generality that $C_{j} (1 \leq j \leq m_{2})$ are 2-clauses and $C_{j} (m_{2}+1 \leq j \leq m)$ are 3-clauses.

For each variable $x_{i}$, we construct a {\em variable gadget}.
It consists of six men $p_{i,1}$, $p_{i,2}$, $p_{i,3}$, $p_{i,4}$, $a_{i,1}$, and $a_{i,2}$, and four women $q_{i,1}$, $q_{i,2}$, $q_{i,3}$, and $q_{i,4}$.
A variable gadget corresponding to $x_{i}$ is called an $x_{i}$-{\em gadget}.
For each clause $C_{j}$, we construct a {\em clause gadget}.
If $C_{j}$ is a 2-clause, we create one man $y_{j}$ and two women $z_{j,1}$ and $z_{j,2}$.
If $C_{j}$ is a 3-clause, we create seven men $y_{j,k}$ ($1 \leq k \leq 7$) and nine women $v_{j,k}$ ($1 \leq k \leq 6$) and $z_{j,k}$ ($1 \leq k \leq 3$).
A clause gadget corresponding to $C_{j}$ is called a $C_{j}$-{\em gadget}.
Additionally, we create a man $s$ and a woman $t$, who constitute a gadget called the {\em separator}.

Thus, there are $6n+m_{2}+7m_{3}+1$ men and $4n+2m_{2}+9m_{3}+1$ women in the created SMTI-instance, denoted $I(f)$.
Finally, we add dummy persons who have empty preference lists to make the numbers of men and women equal.
They do not play any role in the following arguments, so we omit them.

Suppose that $x_{i}$'s $k$th positive occurrence ($k=1, 2$) is in the $d_{i,k}$th clause $C_{d_{i,k}}$ as the $e_{i,k}$th literal ($1 \leq e_{i,k} \leq 3$).
Similarly, suppose that $x_{i}$'s $k$th negative occurrence ($k=1, 2$) is in the $g_{i,k}$th clause $C_{g_{i,k}}$ as the $h_{i,k}$th literal ($1 \leq h_{i,k} \leq 3$).
Then preference lists of ten persons in the $x_{i}$-gadget are constructed as shown in Fig.~\ref{fig:variable}.
Here, each preference list is described as a sequence from left to right according to preference, i.e., the leftmost person is the most preferred and the rightmost person is the least preferred.
Tied persons (i.e., persons with the equal preference) are included in parentheses.
In the figure, both $z_{g_{i,2}, h_{i,2}}$ and $z_{d_{i,2}, e_{i,2}}$ are written but actually either one is null depending on which polarity of $x_{i}$ occurs once. 
Men are aligned in the order of $p_{i,1}$,  $p_{i,3}$, $a_{i,1}$, $a_{i,2}$, $p_{i,2}$, and $p_{i,4}$ from top to bottom, and women are aligned in the order of $q_{i,1}$, $q_{i,3}$, $q_{i,2}$, and $q_{i,4}$.
(See Fig.~\ref{fig:variable-alignment}. Edges depicted in the figure are those within the variable gadget.)

\begin{figure}[ht]
\begin{center}
\renewcommand\arraystretch{1.4}
\begin{tabular}{lllllllllllllllllllllllllll}
$p_{i,1}$:  & $q_{i,1}$  & $z_{g_{i,1}, h_{i,1}}$ \ & \hspace{15mm} & $q_{i,1}$:  & $a_{i,1}$ & $p_{i,1}$  \\
$a_{i,1}$:  & ($q_{i,1}$  & $q_{i,2}$) &  \  & $q_{i,2}$:  & $a_{i,1}$ & $p_{i,2}$  \\
$p_{i,2}$:  & $q_{i,2}$  & $z_{d_{i,1}, e_{i,1}}$ \ &  \\
$p_{i,3}$:  & $q_{i,3}$  & $z_{g_{i,2}, h_{i,2}}$ \ &  & $q_{i,3}$:  & $a_{i,2}$ & $p_{i,3}$  \\
$a_{i,2}$:  & ($q_{i,3}$  & $q_{i,4}$) &  \ \hspace{15mm} & $q_{i,4}$:  & $a_{i,2}$ & $p_{i,4}$  \\
$p_{i,4}$:  & $q_{i,4}$  & $z_{d_{i,2}, e_{i,2}}$ \ & \\
\end{tabular}
\caption{Preference lists of persons in $x_{i}$-gadget.}\label{fig:variable}
\end{center}
\end{figure}

\begin{figure}[ht]
  \centering
  \includegraphics[width=3cm]{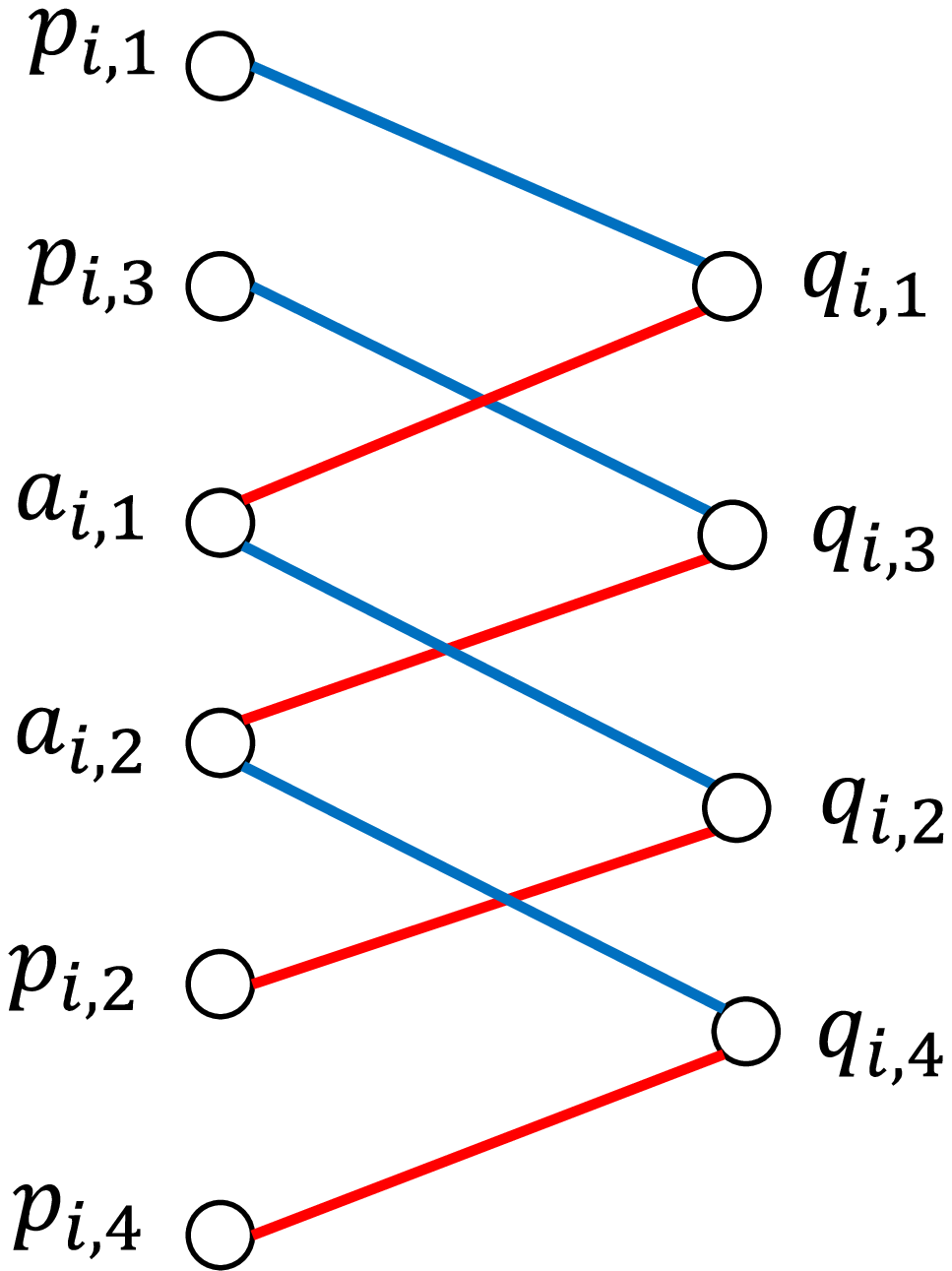}
  \caption{Alignment of agents in a variable gadget.}\label{fig:variable-alignment}
\end{figure}

It might be helpful to explain here intuition behind a variable gadget.
People there are partitioned into two groups, $\{ p_{i,1}, a_{i,1}, p_{i,2}, q_{i,1}, q_{i,2} \}$ and $\{ p_{i,3}, a_{i,2},  p_{i,4}, q_{i,3}, q_{i,4} \}$.
The first group  corresponds to the first positive occurrence and the first negative occurrence of $x_{i}$.
It has two stable matchings $\{ (p_{i,1}, q_{i,1}), (a_{i,1}, q_{i,2}) \}$ (blue in Fig.~\ref{fig:variable-alignment}) and $\{ (a_{i,1}, q_{i,1}),  (p_{i,2}, q_{i,2}) \}$ (red).
We associate the former with the assignment $x_{i}=0$ and the latter with the assignment $x_{i}=1$.
The second group corresponds to the second positive occurrence and the second negative occurrence of $x_{i}$ (if any).
It has two stable matchings $\{ (p_{i,3}, q_{i,3}), (a_{i,2}, q_{i,4}) \}$ (blue) and $\{ (a_{i,2}, q_{i,3}),  (p_{i,4}, q_{i,4}) \}$ (red). 
We associate the former with $x_{i}=0$ and the latter with $x_{i}=1$.
Entanglement of two groups in Fig.~\ref{fig:variable-alignment} plays a role of ensuring consistency of assignments between the first and the second groups.
Depending on the choice of the matching in the first group, edges with the same color must be chosen from the second group to avoid edge-crossing.

Let us continue the reduction.
We then construct preference lists of clause gadgets.
Consider a clause $C_{j}$, and suppose that its $k$th literal is of a variable $x_{j_{k}}$.
Define $\ell_{j,k}$ as
\begin{equation}
\displaystyle \ell_{j,k}   = \\
\begin{cases}
\mbox{1 if this is the 1st negative occurrence of  $x_{j_{k}}$} \\ \nonumber
\mbox{2 if this is the 1st positive occurrence of  $x_{j_{k}}$} \\ \nonumber
\mbox{3 if this is the 2nd negative occurrence of  $x_{j_{k}}$} \\ \nonumber
\mbox{4 if this is the 2nd positive occurrence of  $x_{j_{k}}$}.
\end{cases}
\end{equation}
If $C_{j}$ is a 2-clause (respectively, 3-clause), then the preference lists of persons in the $C_{j}$-gadget are as shown in Fig.~\ref{fig:clause2} (respectively, Fig.~\ref{fig:clause3}).
The alignment order of persons in each clause gadget is the same as in Figs.~\ref{fig:clause2} and \ref{fig:clause3}.
Since a clause gadget for a 3-clause is complicated, we show a structure in the leftmost figure of Fig.~\ref{fig:3-clause-alignment} (three matchings $N_{j,1}$, $N_{j,2}$, and $N_{j,3}$ will be used later).

\begin{figure}[ht]
\begin{center}
\renewcommand\arraystretch{1.2}
\begin{tabular}{lllllllllllllllllllllll}
$y_{j}$:  & ($z_{j,1}$ & $z_{j,2}$) & \hspace{15mm} & $z_{j,1}$:  &  $y_{j}$ & $p_{j_{1},\ell_{j,1}}$ \\
 &  &  & & $z_{j,2}$:  &  $y_{j}$ & $p_{j_{2},\ell_{j,2}}$ \\
\end{tabular}
\caption{Preference lists of persons in $C_{j}$-gadget ($1 \leq j \leq m_{2}$).}\label{fig:clause2}
\end{center}
\end{figure}

\begin{figure}[ht]
\begin{center}
\renewcommand\arraystretch{1.2}
\begin{tabular}{lllllllllllllllllllllll}
$y_{j,1}$:  & ($v_{j,1}$ & $v_{j,3}$) & \hspace{15mm} & $v_{j,1}$:  &  $y_{j,1}$ & \\
$y_{j,2}$:  & ($v_{j,2}$ & $z_{j,1}$) & & $v_{j,2}$:  &  $y_{j,2}$ & \\
$y_{j,3}$:  & ($v_{j,3}$ & $v_{j,4}$) & & $v_{j,3}$:  &  $y_{j,1}$ & $y_{j,3}$ \\
$y_{j,4}$:  & ($z_{j,2}$ & $v_{j,5}$) & & $z_{j,1}$:  &  $y_{j,2}$ & $p_{j_{1},\ell_{j,1}}$ \\
$y_{j,5}$:  & ($v_{j,4}$ & $v_{j,6}$) & & $z_{j,2}$:  &  $y_{j,4}$ & $p_{j_{2},\ell_{j,2}}$  \\
$y_{j,6}$:  & ($v_{j,5}$ & $z_{j,3})$ & & $v_{j,4}$:  &  $y_{j,5}$ & $y_{j,3}$  \\
$y_{j,7}$:  & $v_{j,6}$ & & & $v_{j,5}$:  &  $y_{j,6}$ & $y_{j,4}$  \\
 &  &  & & $v_{j,6}$:  &  $y_{j,5}$ & $y_{j,7}$ \\
 &  &  & & $z_{j,3}$:  &  $y_{j,6}$ & $p_{j_{3},\ell_{j,3}}$ \\
\end{tabular}
\caption{Preference lists of persons in $C_{j}$-gadget ($m_{2}+1 \leq j \leq m$).}\label{fig:clause3}
\end{center}
\end{figure}

\begin{figure}[ht]
  \centering
  \includegraphics[width=11.5cm]{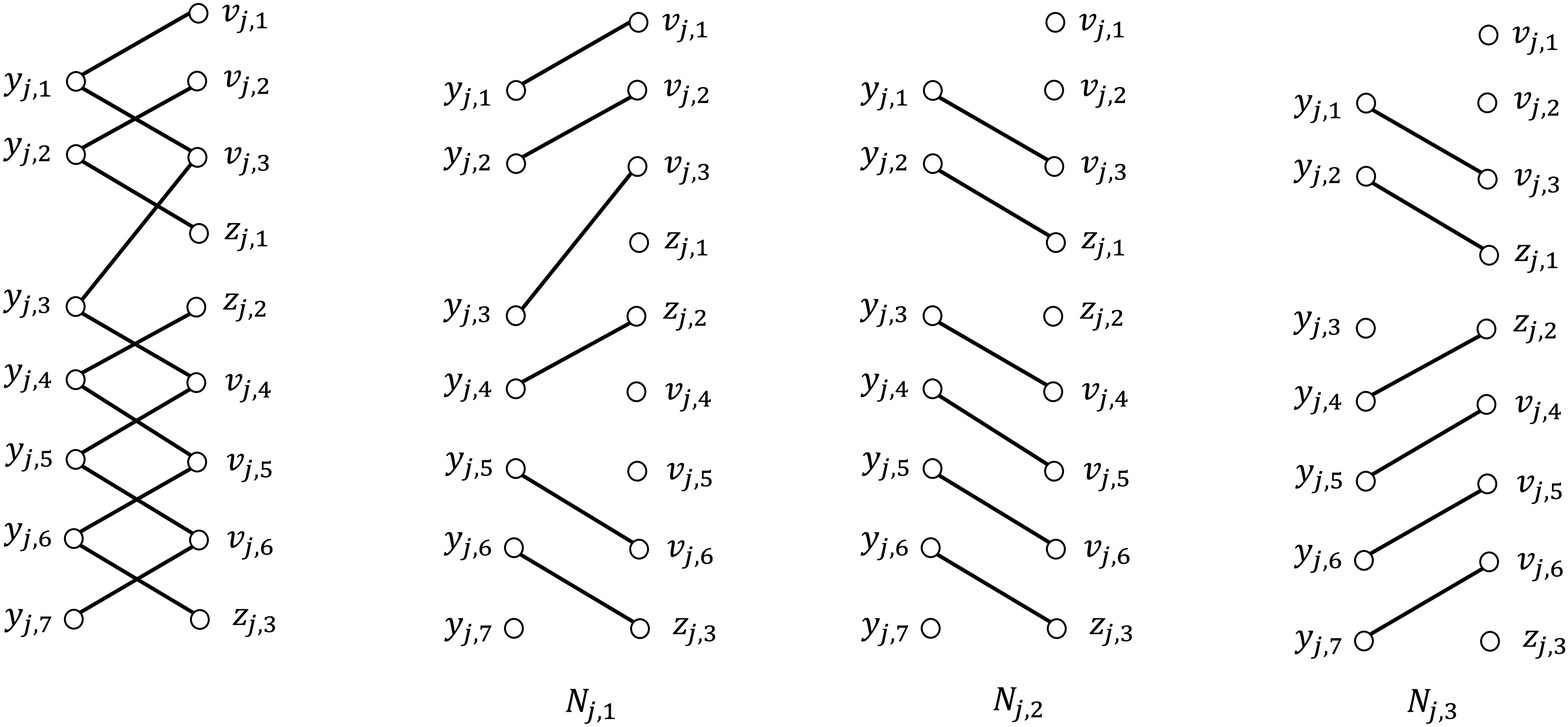}
  \caption{Acceptability graph of a 3-clause gadget $C_{j}$ and its matchings $N_{j,1}$, $N_{j,2}$, and $N_{j,3}$.}\label{fig:3-clause-alignment}
\end{figure}

Finally, each of the man and the woman in the separator includes only the other in the list (Fig.~\ref{fig:separator}).  They are guaranteed to be matched together in any stable matching.

\begin{figure}[ht]
\begin{center}
\renewcommand\arraystretch{1.2}
\begin{tabular}{llllll}
$s$: \ & $t$ & \hspace{15mm} & $t$: \ &  $s$
\end{tabular}
\caption{Preference lists of the man and the woman in the separator.}\label{fig:separator}
\end{center}
\end{figure}

Alignment of the whole instance is depicted in Fig.~\ref{fig:alignment}.
Variable gadgets are placed top, then followed by the separator, clause gadgets come bottom.
The separator plays a role of prohibiting a person of a variable gadget to be matched with a person of a clause gadget; if they are matched, then the corresponding edge crosses the separator.

\begin{figure}[ht]
  \centering
  \includegraphics[width=3.2cm]{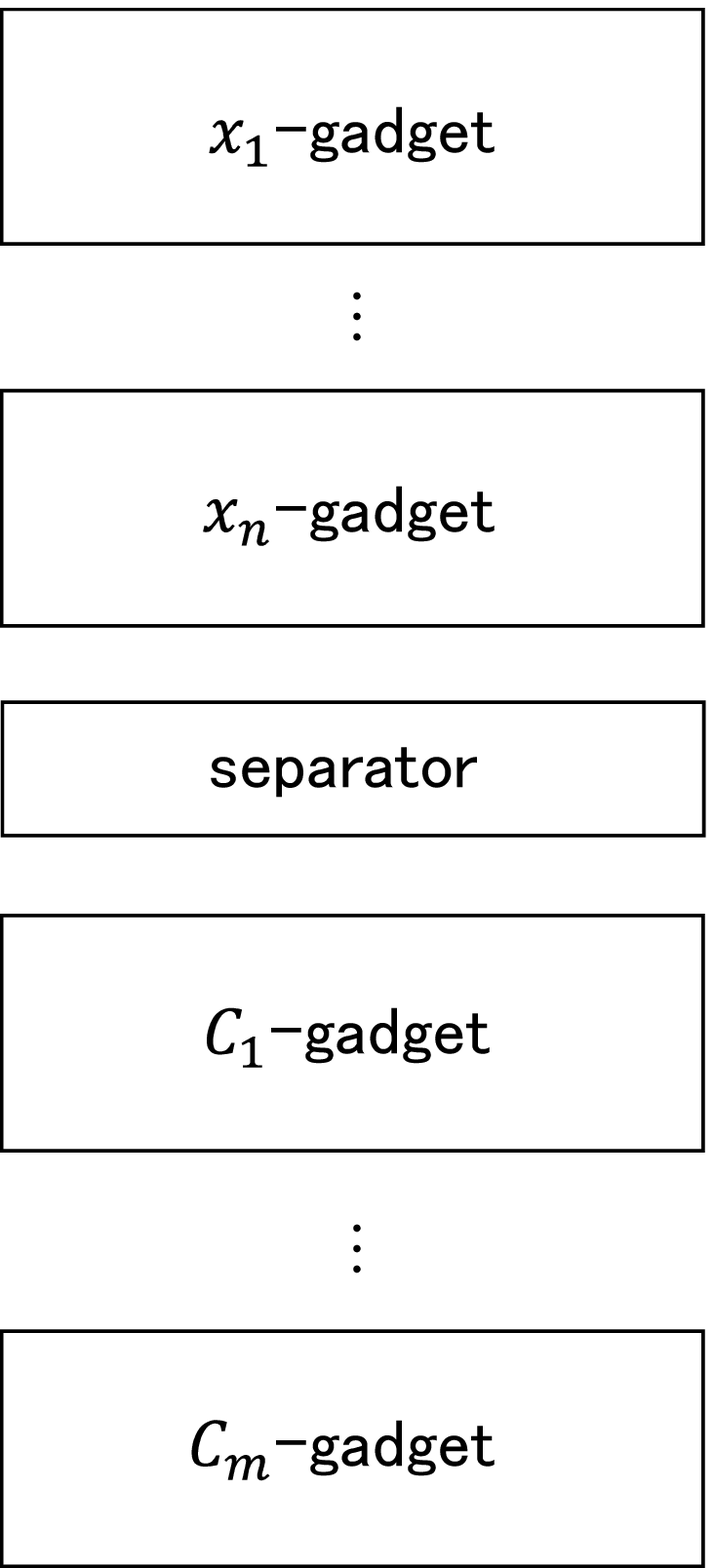}
  \caption{Alignment of agents.}\label{fig:alignment}
\end{figure}

Now the reduction is completed.
It is not hard to see that the reduction can be performed in polynomial time and the conditions on the preference lists stated in the theorem are satisfied.

We then show the correctness.
First, suppose that $f$ is satisfiable and let $A$ be a satisfying assignment.  
We construct a weak-SSNM $M$ of $I(f)$ from $A$.
For an $x_{i}$-gadget, define two matchings 

\begin{itemize}

\item $M_{i,0}= \{ (p_{i,1}, q_{i,1}), (a_{i,1}, q_{i,2}), (p_{i,3}, q_{i,3}), (a_{i,2}, q_{i,4}) \}$ (blue in Fig.~\ref{fig:variable-alignment}) and 

\item $M_{i,1}= \{ (a_{i,1}, q_{i,1}), (p_{i,2}, q_{i,2}), (a_{i,2}, q_{i,3}), (p_{i,4}, q_{i,4}) \}$ (red in Fig.~\ref{fig:variable-alignment}).

\end{itemize}

\noindent
If $x_{i}=0$ under $A$, then add $M_{i,0}$ to $M$; otherwise, add $M_{i,1}$ to $M$.
For a $C_{j}$-gadget ($1 \leq j \leq m_{2}$), we define two matchings $N_{j,1}= \{ (y_{j}, z_{j,2}) \}$ and $N_{j,2}= \{ (y_{j}, z_{j,1}) \}$.
For a $C_{j}$-gadget ($m_{2}+1 \leq j \leq m$), we define three matchings

\begin{itemize}
\item $N_{j,1}= \{ (y_{j,1}, v_{j,1}), (y_{j,2}, v_{j,2}), (y_{j,3}, v_{j,3}), (y_{j,4}, z_{j,2}), (y_{j,5}, v_{j,6}), (y_{j,6}, z_{j,3}) \}$, 

\item $N_{j,2}= \{ (y_{j,1}, v_{j,3}), (y_{j,2}, z_{j,1}), (y_{j,3}, v_{j,4}), (y_{j,4}, v_{j,5}), (y_{j,5}, v_{j,6}), (y_{j,6}, z_{j,3}) \}$, and 

\item $N_{j,3}= \{ (y_{j,1}, v_{j,3}), (y_{j,2}, z_{j,1}), (y_{j,4}, z_{j,2}), (y_{j,5}, v_{j,4}), (y_{j,6}, v_{j,5}), (y_{j,7}, v_{j,6}) \}$,
\end{itemize}

\noindent
that are depicted in Fig.~\ref{fig:3-clause-alignment}.
Note that, for each $k(=1, 2, 3)$, only $z_{j,k}$ (among $z_{j,1}$, $z_{j,2}$, and $z_{j,3}$) is single in $N_{j,k}$.
If $C_{j}$ is satisfied by the $k$th literal ($k=1, 2, 3$), then add $N_{j,k}$ to $M$.
(If $C_{j}$ is satisfied by more than one literal, then choose one arbitrarily.)
Finally add the pair $(s, t)$ to $M$.

It is not hard to see that $M$ is noncrossing.
We show that it is weakly stable.
Clearly, neither $s$ nor $t$ in the separator forms a blocking pair.
Next, consider the $x_{i}$-gadget.
In $M_{i,0}$, women $q_{i,2}$ and $q_{i,4}$ are matched with the first-choice man.
The woman $q_{i,1}$ is matched with the second-choice man $p_{i,1}$ but her first-choice man $a_{i,1}$ is matched with a first-choice woman $q_{i,2}$.
Similarly, $q_{i,3}$'s first-choice man $a_{i,2}$ is matched with a first-choice woman $q_{i,4}$.
Men $p_{i,1}$, $a_{i,1}$, $p_{i,3}$, and $a_{i,2}$ are matched with a first-choice woman.
Hence these persons cannot be a part of a blocking pair; only $p_{i,2}$ and $p_{i,4}$ may participate in a blocking pair.
Similarly, we can argue that, in $M_{i,1}$, only $p_{i,1}$ and $p_{i,3}$ may participate in a blocking pair.

For a $C_{j}$-gadget ($1 \leq j \leq m_{2}$), in either $N_{j,1}$ and $N_{j,2}$, both the matched persons obtain a first-choice partner.
Consider a $C_{j}$-gadget ($m_{2}+1 \leq j \leq m$).
In $N_{j,1}$, all the men except for $y_{j,7}$ are matched with a first-choice woman.
$y_{j,7}$'s unique choice $v_{j,6}$ is matched with the first-choice man $y_{j,5}$.
Hence no man in this gadget can participate in a blocking pair, and so no blocking pair exists within this gadget.
Since $z_{j,2}$ and $z_{j,3}$ are matched with their respective first-choice woman, only the possibility is that $z_{j,1}$ forms a blocking pair with $p_{j_{1},\ell_{j,1}}$ of a variable gadget.
The same observation applies for $N_{j,2}$ and $N_{j,3}$ and we can see that for each $k(=1, 2, 3)$ only $z_{j,k}$ can participate in a blocking pair in $N_{j,k}$.

To summarize, if there exists a blocking pair, it must be of the form $(p_{i,\ell}, z_{j,k})$ for some $i, \ell, j$, and $k$, and both $p_{i,\ell}$ and $z_{j,k}$ are single in $M$.
Suppose that $\ell=1$.
The reason for $(p_{i,1}, z_{j,k})$ being an acceptable pair is that $C_{j}$'s $k$th literal is $\lnot{x_{i}}$, a negative occurrence of $x_{i}$.
Since $p_{i,1}$ is single, $M_{i,1} \subset M$ and hence $x_{i}=1$ under $A$.
Since $z_{j,k}$ is single, $N_{j,k} \subset M$ and hence $C_{j}$ is satisfied by its $k$th literal $\lnot{x_{i}}$, but this is a contradiction.
The other cases $\ell=2, 3, 4$ can be argued in the same manner, and we can conclude that $M$ is stable.

Conversely, suppose that $I(f)$ admits a weak-SSNM $M$.
We construct a satisfying assignment $A$ of $f$.
Before giving construction, we observe structural properties of $M$ in two lemmas:

\begin{lemma}\label{lemma:variable}
For each $i$ ($1 \leq i \leq n$), either $M_{i,0} \subset M$ or $M_{i,1} \subset M$.
\end{lemma}

\begin{proof}
Note that preference lists of the ten persons of the $x_{i}$-gadget include persons of the same $x_{i}$-gadget or some persons from clause gadgets.
Hence, due to the separator, persons of the $x_{i}$-gadget can only be matched within this gadget to avoid edge-crossings.

Note that a stable matching is a maximal matching.
With regard to $p_{i,1}$, $a_{i,1}$, $p_{i,2}$, $q_{i,1}$, and $q_{i,2}$, there are three maximal matchings $\{ (p_{i,1}, q_{i,1}), (a_{i,1}, q_{i,2}) \}$, $\{ (a_{i,1}, q_{i,1}),  (p_{i,2}, q_{i,2}) \}$, and $\{ (p_{i,1}, q_{i,1}),  (p_{i,2}, q_{i,2}) \}$, but the last one is blocked by $ (a_{i,1}, q_{i,1})$ and $(a_{i,1}, q_{i,2})$.
Hence either the first or the second one must be in $M$.
With regard to $p_{i,3}$, $a_{i,2}$, $p_{i,4}$, $q_{i,3}$, and $q_{i,4}$, there are three maximal matchings $\{ (p_{i,3}, q_{i,3}), (a_{i,2}, q_{i,4}) \}$, $\{ (a_{i,2}, q_{i,3}), (p_{i,4}, q_{i,4}) \}$, and $\{ (p_{i,3}, q_{i,3}),  (p_{i,4}, q_{i,4}) \}$, but the last one is blocked by $ (a_{i,2}, q_{i,3})$ and $(a_{i,2}, q_{i,4})$.
Hence either the first or the second one must be in $M$.

If we choose $\{ (p_{i,1}, q_{i,1}), (a_{i,1}, q_{i,2}) \}$, then we must choose $\{ (p_{i,3}, q_{i,3}), (a_{i,2}, q_{i,4}) \}$ to avoid edge-crossing, which constitute $M_{i,0}$.
If we choose $\{ (a_{i,1}, q_{i,1}), (p_{i,2}, q_{i,2}) \}$, then we must choose $\{ (a_{i,2}, q_{i,3}), (p_{i,4}, q_{i,4}) \}$, which constitute $M_{i,1}$.
Hence either $M_{i,0}$ or $M_{i,1}$ must be a part of $M$.
\qed
\end{proof}

\begin{lemma}\label{lemma:clause}
(i) For a $C_{j}$-gadget ($1 \leq j \leq m_{2}$), at least one of $z_{j,1}$ and $z_{j,2}$ is unmatched in $M$.
(ii) For a $C_{j}$-gadget ($m_{2}+1 \leq j \leq m$), at least one of $z_{j,1}$, $z_{j,2}$, and $z_{j,3}$ is unmatched in $M$.
\end{lemma}

\begin{proof}
(i) Note that preference lists of the three persons of the $C_{j}$-gadget include persons of the same $C_{j}$-gadget or some persons from variable gadgets.
To avoid edge-crossing, persons must be matched within the same $C_{j}$-gadget.
Then it is impossible that both $z_{j,1}$ and $z_{j,2}$ are matched in $M$.

(ii) Again, we note that people must be matched within the same $C_{j}$-gadget.
For contradiction, suppose that all $z_{j,1}$, $z_{j,2}$, and $z_{j,3}$ are matched in $M$.
Then $(y_{j,2}, z_{j,1})$, $(y_{j,4}, z_{j,2})$, and $(y_{j,6}, z_{j,3})$ are in $M$ (Fig.~\ref{fig:proofsup}(1)).
To avoid edge-crossing, $(y_{j,3}, v_{j,3})$, $(y_{j,3}, v_{j,4})$, and $(y_{j,7}, v_{j,6})$ must not be in $M$ (Fig.~\ref{fig:proofsup}(2)).
The pair $(y_{j,5}, v_{j,4})$ must be in $M$ as otherwise $(y_{j,3}, v_{j,4})$ is a blocking pair (Fig.~\ref{fig:proofsup}(3)).
For $M$ to be a matching, $(y_{j,5}, v_{j,6})$ must not be in $M$ (Fig.~\ref{fig:proofsup}(4)).
Then  $(y_{j,7}, v_{j,6})$ is a blocking pair, a contradiction.
\qed
\end{proof}

\begin{figure}[ht]
  \centering
  \includegraphics[width=11.5cm]{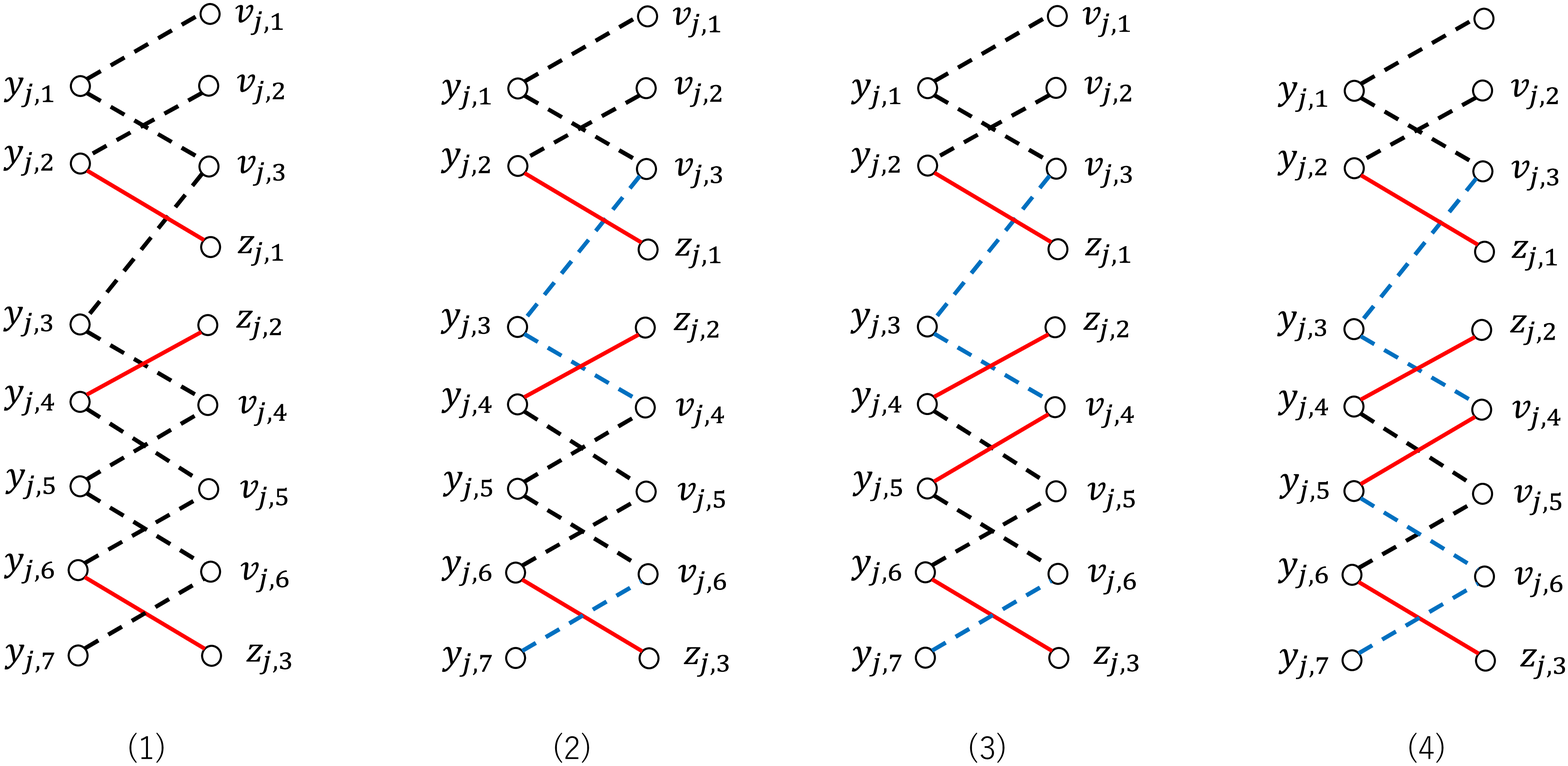}
  \caption{Situation in the proof of Lemma \ref{lemma:clause}(ii). Red solid edges are those confirmed to be in $M$, blue dashed edges are those confirmed not to be in $M$, and black dashed edges are uncertain.}\label{fig:proofsup}
\end{figure}

For each $i$, either $M_{i,0} \subset M$ or $M_{i,1} \subset M$ holds by Lemma \ref{lemma:variable}.
If $M_{i,0} \subset M$ holds then we set $x_{i}=0$ in $A$, and if $M_{i,1} \subset M$ holds then we set $x_{i}=1$ in $A$.
We show that $A$ satisfies $f$.
Let $C_{j}$ be an arbitrary clause.
In the following, we assume that $C_{j}$ is a 3-clause, but the same argument holds if $C_{j}$ is a 2-clause.
By Lemma \ref{lemma:clause}(ii), at least one of $z_{j,1}$, $z_{j,2}$, and $z_{j,3}$ is unmatched in $M$.
If there are two or more unmatched women, then choose one arbitrarily and let this woman be $z_{j,k}$.
We show that $C_{j}$ is satisfied by its $k$th literal.
Suppose not.

First suppose that the $k$th literal of $C_{j}$ is the first positive occurrence of $x_{i}$.
Then, by construction of preference lists, $(p_{i,2}, z_{j,k})$ is an acceptable pair.
If $x_{i}=0$ under $A$, then $M_{i,0} \subset M$ by construction of $A$, and hence $p_{i,2}$ is single in $M$.
Thus $(p_{i,2}, z_{j,k})$ is a blocking pair, which contradicts stability of $M$.
Hence $x_{i}=1$ under $A$ and $C_{j}$ is satisfied by $x_{i}$.
When the $k$th literal of $C_{j}$ is the second positive occurrence of $x_{i}$, the same argument holds if we replace $p_{i,2}$ by $p_{i,4}$.

Next suppose that the $k$th literal of $C_{j}$ is the first negative occurrence of $x_{i}$.
Then, by construction of preference lists, $(p_{i,1}, z_{j,k})$ is an acceptable pair.
If $x_{i}=1$ under $A$, then $M_{i,1} \subset M$ by construction of $A$, and hence $p_{i,1}$ is single in $M$.
Thus $(p_{i,1}, z_{j,k})$ is a blocking pair, which contradicts stability of $M$.
Hence $x_{i}=0$ under $A$ and $C_{j}$ is satisfied by $\lnot x_{i}$.
If the $k$th literal of $C_{j}$ is the second negative occurrence of $x_{i}$, the same argument holds if we replace $p_{i,1}$ by $p_{i,3}$.
Thus $A$ is a satisfying assignment of $f$ and the proof is completed.
\qed
\end{proof}

Next we give a positive result.

\begin{theorem}\label{thm:SSNM-weak-positive}
The problem of determining if a weak-SSNM exists, given an SMTI-instance, is solvable in $O(n)$-time if each man's preference list contains at most one woman.
\end{theorem}

\begin{proof}
Let $I$ be an input SMTI-instance.
First, we construct a bipartite graph $G_{I}=(U_{I}, V_{I}, E_{I})$, where $U_{I}$ and $V_{I}$ correspond to the sets of men and women in $I$, respectively, and $(m,w) \in E_{I}$ if and only if $m$ is a first-choice of $w$.
For a vertex $v \in V_{I}$, let $d(v)$ denote its degree in $G_{I}$.
Since acceptability is mutual, if a woman $w$'s preference list in $I$ is nonempty, $d(w) \geq 1$ holds.
Note that it can happen that $d(w) \geq 2$ because preference lists may contain ties.
In the following lemma, we characterize (not necessarily noncrossing) stable matchings of $I$. 

\begin{lemma}\label{lemma:3}
$M$ is a stable matching of $I$ if and only if $M \subseteq E_{I}$ and each woman $w \in V_{I}$ such that $d(w) \geq 1$ is matched in $M$.
\end{lemma}

\begin{proof}
Suppose that $M$ is stable.  
If $M \not\subseteq E_{I}$, there is an edge $(m,w) \in M \setminus E_{I}$.
The fact $(m,w) \not\in E_{I}$ means that $m$ is not $w$'s first-choice so there is an edge $(m', w) \in E_{I}$ such that $m' \succ_{w} m$.
Since $(m, w) \in M$, $m'$ is single in $M$.
Therefore, $(m', w)$ is a blocking pair for $M$, a contradiction.
If there is a woman $w \in V_{I}$ such that $d(w) \geq 1$ but $w$ is single in $M$, then any man $m$ such that $(m,w)$ is an acceptable pair is a blocking pair because $m$ is also single in $M$, a contradiction.

Conversely, suppose that $M$ satisfies the conditions of the right hand side.
Then each woman who has a nonempty list is matched with a first-choice man, so there cannot be a blocking pair.
\qed
\end{proof}

By Lemma \ref{lemma:3}, our task is to select from $E_{I}$ one edge per woman $w$ such that $d(w) \geq 1$, in such a way that the resulting matching is noncrossing.
We do this greedily.
$M$ is initially empty, and we add edges to $M$ by processing vertices of $V_{I}$ from top to bottom.
At $w_{i}$'s turn, if $d(w_{i}) \geq 1$, then choose the topmost edge that does not cross any edge in $M$, and add it to $M$.
If there is no such edge, then we immediately conclude that $I$ admits no weak-SSNM.
If we can successfully process all the women, we output the final matching $M$, which is a weak-SSNM.

In the following, we formalize the above idea.
A pseudo-code of the whole algorithm {\sc Weak-SSNM-1} is given in Algorithm~\ref{alg:1}.

\begin{algorithm}[htb]
  \caption{{\sc Weak-SSNM-1}}\label{alg:1}
  \begin{algorithmic}[1]
   \REQUIRE An SMTI-instance $I$.
   \ENSURE A weak-SSNM $M$ or ``No'' if none exists.
   \STATE Construct a bipartite graph $G_{I}=(U_{I}, V_{I}, E_{I})$.\label{step:3}
   \STATE Let $M:= \emptyset$.
   \FOR{$i=1$ to $n$}
     \IF{$d(w_{i}) \geq 1$}
       \STATE Let $j^{*}$ (if any) be the smallest $j$ such that $(m_{j}, w_{i}) \in E_{I}$ and $M \cup \{(m_{j}, w_{i}) \}$ is a noncrossing matching.\label{step:10}
       \STATE Let $M := M \cup \{ (m_{j^{*}}, w_{i}) \}$.
            \IF{no such $j^{*}$ exists}
              \STATE Output ``No'' and halt.
            \ENDIF
     \ENDIF
   \ENDFOR
   \STATE Output $M$.
  \end{algorithmic}
 \end{algorithm}

We show the correctness.
Suppose that {\sc Weak-SSNM-1} outputs a matching $M$.
$M$ is noncrossing by the condition of line \ref{step:10}, and $M$ is stable because the construction of $M$ follows the condition of Lemma \ref{lemma:3}.

Conversely, suppose that $I$ admits a weak-SSNM $M^{*}$.
We show that {\sc Weak-SSNM-1} outputs a matching.
Suppose not, and suppose that {\sc Weak-SSNM-1} failed when processing woman (vertex) $w_{k}$.
Let $\bar{M}$ be the matching constructed so far by {\sc Weak-SSNM-1}.
Then for each $i$ ($1 \leq i \leq k-1$), $w_{i}$ is single in $M^{*}$ if and only if she is single in $\bar{M}$.
Also, since $M^{*} \subseteq E_{I}$ by Lemma \ref{lemma:3}, we can show by a simple induction that for each $i$ ($1 \leq i \leq k-1$), if $M^{*}(w_{i})=m_{p}$ and $\bar{M}(w_{i})=m_{q}$, then $q \leq p$.
Then, at line \ref{step:10}, we could have chosen $(M^{*}(w_{k}), w_{k})$ to add to $\bar{M}$, a contradiction.

Finally, we consider time-complexity.
Since the preference list of each man contains at most one woman, the graph $G_{I}$ at line \ref{step:3} can be constructed in $O(n)$-time and contains at most $n$ edges.
The {\bf for}-loop can be executed in $O(n)$-time because each edge is scanned at most once in the loop; whether or not an edge crosses edges of $M$ at line \ref{step:10} can be done in constant time by keeping the maximum index of the matched men in $M$ at any stage.
\qed
\end{proof}

\section{Maximum Cardinality Weakly Stable Noncrossing Matchings}\label{sec:WSNM}

In this section, we present an algorithm to find a maximum cardinality WSNM.
For an instance $I$, let $opt(I)$ denote the size of the maximum cardinality WSNM.

\subsection{SMI}\label{sec:WSNM-SMI}

Let $I'$ be a given instance with men $m_{1}, \ldots, m_{n}$ and women $w_{1}, \ldots, w_{n}$.
To simplify the description of the algorithm, we translate $I'$ to an instance $I$ by adding a man $m_{0}$ and a woman $w_{0}$, each of whom includes only the other in the preference list, and similarly a man $m_{n+1}$ and a woman $w_{n+1}$, each of whom includes only the other in the preference list.
It is easy to see that, for a WSNM $M'$ of $I'$, $M= M' \cup \{ (m_{0}, w_{0}), (m_{n+1}, w_{n+1}) \}$ is a WSNM of $I$.
Conversely, any WSNM $M$ of $I$ includes the pairs $(m_{0}, w_{0})$ and $(m_{n+1}, w_{n+1})$, and $M' = M \setminus \{ (m_{0}, w_{0}), (m_{n+1}, w_{n+1}) \}$ is a WSNM of $I'$.
Thus we have that $opt(I)=opt(I')+2$.
Hence, without loss of generality, we assume that a given instance $I$ has $n+2$ men and $n+2$ women, with $m_{0}$, $w_{0}$, $m_{n+1}$, and $w_{n+1}$ having the above mentioned preference lists.

Let $M = \{ (m_{i_{1}}, w_{j_{1}}), (m_{i_{2}}, w_{j_{2}}), \ldots, (m_{i_{k}}, w_{j_{k}}) \}$ be a noncrossing matching of $I$ such that $i_{1} < i_{2} \cdots < i_{k}$ and $j_{1} < j_{2} \cdots < j_{k}$.  We call $(m_{i_{k}}, w_{j_{k}})$ the {\em maximum pair} of $M$.
Suppose that $(m_{x}, w_{y})$ is the maximum pair of a noncrossing matching $M$.
We call $M$ a {\em semi-WSNM} if each of its noncrossing blocking pair $(m_{i}, w_{j})$ (if any) satisfies $x \leq i \leq n+1$ and $y \leq j \leq n+1$.
Intuitively, a semi-WSNM is a WSNM up to its maximum pair.
Note that any semi-WSNM must contain $(m_{0}, w_{0})$, as otherwise it is a noncrossing blocking pair.
For $0 \leq i \leq n+1$ and $0 \leq j \leq n+1$, we define $X(i,j)$ as the maximum size of a semi-WSNM of $I$ whose maximum pair is $(m_{i}, w_{j})$; if $I$ does not admit a semi-WSNM with the maximum pair $(m_{i}, w_{j})$, $X(i,j)$ is defined to be $-\infty$.

\begin{lemma}\label{lm:1}
$opt(I) = X(n+1, n+1)$.
\end{lemma}

\begin{proof}
Note that any WSNM of $I$ includes $(m_{n+1}, w_{n+1})$, as otherwise it is a noncrossing blocking pair.
Hence it is a semi-WSNM with the maximum pair $(m_{n+1}, w_{n+1})$.
Conversely, any semi-WSNM with the maximum pair $(m_{n+1}, w_{n+1})$ does not include a noncrossing blocking pair and hence is also a WSNM.
Therefore, the set of WSNMs is equivalent to the set of semi-WSNMs with the maximum pair $(m_{n+1}, w_{n+1})$.  This completes the proof.
\qed
\end{proof}

To compute $X(n+1, n+1)$, we shortly define quantity $Y(i,j)$ ($0 \leq i \leq n+1, 0 \leq j \leq n+1$) using recursive formulas, and show that $Y(i,j)=X(i,j)$ for all $i$ and $j$.
We then show that these recursive formulas allow us to compute $Y(i,j)$ in polynomial time using dynamic programming.

We say that two noncrossing edges $(m_{i}, w_{j})$ and $(m_{x}, w_{y})$ ($i < x, j < y$) are {\em conflicting} if they contain a noncrossing blocking pair between them; precisely speaking, $(m_{i}, w_{j})$ and $(m_{x}, w_{y})$ are conflicting if the matching $\{ (m_{i}, w_{j}), (m_{x}, w_{y}) \}$ contains a blocking pair $(m_{s}, w_{t})$ such that $i \leq s \leq x$ and $j \leq t \leq y$.
Otherwise, they are {\em nonconflicting}.
Intuitively, two conflicting edges cannot be consecutive elements of a semi-WSNM.

Now $Y(i,j)$ is defined in the Equations (1)--(4).
For convenience, we assume that $-\infty+1=-\infty$.
In Equation (\ref{eq:ij}), $Y(i',j')$ in $\max\{\}$ is taken among all $(i', j')$ such that (i) $0 \leq i' \leq i-1$, (ii) $0 \leq j' \leq j-1$, (iii) $(m_{i'},w_{j'})$ is an acceptable pair, and (iv) $(m_{i},w_{j})$ and $(m_{i'},w_{j'})$ are nonconflicting.
If no such $(i', j')$ exists, $\max\{ Y(i', j') \}$ is defined as $-\infty$ and as a result $Y(i,j)$ is also computed as $-\infty$.

\begin{equation}
Y(0,0)  = 1 \\ \label{eq:00}
\end{equation}

\begin{equation}
Y(0,j)  = -\infty \hspace{5mm} (1\leq j \leq n+1)  \\ \label{eq:0j}
\end{equation}

\begin{equation}
Y(i,0)  = -\infty \hspace{5mm} (1\leq i \leq n+1)  \\ \label{eq:i0}
\end{equation}

\begin{equation}
\displaystyle Y(i,j)   = \\
\begin{cases}
\displaystyle 1+\max\{ Y(i', j') \} \ \ \text{(if $(m_{i}, w_{j})$ is an acceptable pair)} \\ \nonumber
-\infty  \ \  \text{(otherwise)}
\end{cases}
\end{equation}
\begin{equation}
\hspace{40mm} (1 \leq i \leq n+1, 1 \leq j \leq n+1) \label{eq:ij}
\end{equation}

\begin{lemma}\label{lm:2}
$Y(i,j)=X(i,j)$ for any $i$ and $j$ such that $0 \leq i \leq n+1$ and $0 \leq j \leq n+1$.
\end{lemma}

\begin{proof}
We prove the claim by induction.
We first show that $Y(0,0)=X(0,0)$.
The matching $\{ (m_{0},w_{0}) \}$ is the unique semi-WSNM with the maximum pair $(m_{0},w_{0})$, so $X(0,0)=1$ by definition.
Also, $Y(0,0)=1$ by Equation (\ref{eq:00}).  Hence we are done.
We then show that $Y(0,j)=X(0,j)$ for $1 \leq j \leq n+1$.
Since $(m_{0}, w_{j})$ is an unacceptable pair, there is no semi-WSNM with the maximum pair $(m_{0}, w_{j})$, so $X(0,j)= -\infty$ by definition.
Also, $Y(0,j)=-\infty$ by Equation (\ref{eq:0j}).  
We can show that $Y(i,0)=X(i,0)$ for $1 \leq i \leq n+1$ by a similar argument.

Next we show that $Y(i,j)=X(i,j)$ holds for $1 \leq i \leq n+1$ and $1 \leq j \leq n+1$.
As an induction hypothesis, we assume that $Y(a,b)=X(a,b)$ holds for $0 \leq a \leq i-1$ and $0 \leq b \leq j-1$.
First, observe that if $X(i,j) \neq -\infty$, then $X(i,j) \geq 2$.
This is because two pairs $(m_{0}, w_{0})$ and $(m_{i}, w_{j})$ must present in any semi-WSNM having the maximum pair $(m_{i}, w_{j})$.

We first consider the case that $X(i,j) \geq 2$.  
Let $X(i,j) =k$.  
Then, there is a semi-WSNM $M$ with the maximum pair $(m_{i},w_{j})$ such that $|M|=k$.
Let $M' = M \setminus \{ (m_{i},w_{j}) \}$ and $(m_{x}, w_{y})$ be the maximum pair of $M'$.
It is not hard to see that $M'$ is a semi-WSNM with the maximum pair $(m_{x}, w_{y})$ and that $|M'|=k-1$.
Therefore, $X(x,y) \geq k-1$ by the definition of $X$, and $Y(x,y) = X(x,y) \geq k-1$ by the induction hypothesis.
Since $M$ is a semi-WSNM, $(m_{i},w_{j})$ and $(m_{x}, w_{y})$ are nonconflicting, so $(x,y)$ satisfies the condition for $(i',j')$ in Equation (\ref{eq:ij}).
Hence $Y(i,j) \geq  1+Y(x,y) \geq k$.
Suppose that $Y(i,j) \geq k+1$.
By the definition of $Y$, this means that there is $(i',j')$ that satisfies conditions (i)--(iv) for Equation (\ref{eq:ij}), and $Y(i',j') \geq k$.
By the induction hypothesis, $X(i',j') = Y(i',j') \geq k$.
Then there is a semi-WSNM $M'$ with the maximum pair $(m_{i'},w_{j'})$ such that $|M'| \geq k$.
Since $M'$ is a semi-WSNM, and $(m_{i'},w_{j'})$ and $(m_{i}, w_{j})$ are nonconflicting, $M=M' \cup \{ (m_{i},w_{j}) \}$ is a semi-WSNM with the maximum pair $(m_{i},w_{j})$ such that $|M|=|M'|+1 \geq k+1$.  
This contradicts the assumption that $X(i,j)=k$.  Hence $Y(i,j) \leq k$ and therefore $Y(i,j) =k$ as desired.

Finally, consider the case that $X(i,j) = -\infty$.
If $(m_{i}, w_{j})$ is unacceptable, then the latter case of Equation (\ref{eq:ij}) is applied and $Y(i,j) = -\infty$.
So assume that $(m_{i}, w_{j})$ is acceptable.
Then the former case of Equation (\ref{eq:ij}) is applied.
It suffices to show that for any $(i',j')$ that satisfies conditions (i)--(iv) (if any), $Y(i',j') = -\infty$ holds.
Assume on the contrary that there is such $(i',j')$ with $Y(i',j') =k$.
Then $X(i',j') =k$ by the induction hypothesis, and there is a semi-WSNM $M'$ such that $|M'|=k$, $(m_{i'},w_{j'})$ is the maximum pair of $M'$, and $(m_{i'},w_{j'})$ and $(m_{i}, w_{j})$ are nonconflicting.
Then $M=M' \cup \{ (m_{i},w_{j}) \}$ is a semi-WSNM such that $(m_{i},w_{j})$ is the maximum pair and $|M|=|M'|+1 = k+1$, implying that  $X(i,j)=k+1$.
This contradicts the assumption that $X(i,j)=-\infty$ and the proof is completed.
\qed
\end{proof}

Now we analyze time-complexity of the algorithm.
Computing each $Y(0,0)$, $Y(0,j)$, and $Y(i,0)$ can be done in constant time.
For computing one $Y(i,j)$ according to Equation (\ref{eq:ij}), there are $O(n^{2})$ candidates for $(i',j')$.
For each $(i',j')$, checking if $(m_{i'},w_{j'})$ and $(m_{i},w_{j})$ are conflicting can be done in constant time with $O(n^{4})$-time preprocessing described in subsequent paragraphs.
Therefore one $Y(i,j)$ can be computed in time $O(n^{2})$.
Since there are $O(n^{2})$ $Y(i,j)$s, the time-complexity for computing all $Y(i,j)$s is $O(n^{4})$. 
Adding the $O(n^{4})$-time for preprocessing mentioned above, the total time-complexity of the algorithm is $O(n^{4})$. 

In the preprocessing, we construct three tables $S$, $A$, and $B$.
\begin{itemize}
\item $S$ is a $\Theta(n^{4})$-sized four-dimensional table that takes logical values 0 and 1.
For $0 \leq i' \leq i \leq n+1$ and $0 \leq j' \leq j \leq n+1$, $S(i',i,j',j) =1$ if and only if there exists at least one acceptable pair $(m,w)$ such that $m \in \{m_{i'}, m_{i'+1}, \dots, m_{i}\}$ and $w \in \{w_{j'}, w_{j'+1}, \dots, w_{j}\}$.
Since $S(i,i,j,j)=1$ if and only if $(m_{i}, w_{j})$ is an acceptable pair, it can be computed in constant time.
In general, $S(i',i,j',j)$ can be computed in constant time as follows.
\[
 S(i',i,j',j)=
 \begin{cases}
   1 \hspace{5mm} \text{(if $(m_{i}, w_{j})$ is an acceptable pair)}\\
   S(i',i-1,j',j) \lor S(i',i,j',j-1) \hspace{5mm} \text{(otherwise)}
 \end{cases}
\]
Hence $S$ can be constructed in $O(n^{4})$ time by a simple dynamic programming.

\item $A$ is a $\Theta(n^{3})$-sized table.
For convenience, we introduce an imaginary person $\lambda$ who is acceptable to any person, where $q \succ_{p} \lambda$ holds for any person $p$ and any person $q$ acceptable to $p$.
For $0 \leq i \leq n+1$ and $0 \leq j' \leq j \leq n+1$, $A(i,j',j)$ stores the woman whom $m_{i}$ most prefers among $\{w_{j'},\dots,w_{j}, \lambda\}$.
Then, for $i$ and $j$, $A(i,j,j) = w_{j}$ if $(m_{i}, w_{j})$ is an acceptable pair and $A(i,j,j) = \lambda$ otherwise.
$A(i,j',j)$ can be computed as the better of $A(i,j',j-1)$ and $A(i,j,j)$ in $m_{i}$'s list.
By the above arguments, each element of $A$ can be computed in constant time and hence $A$ can be constructed in $O(n^{3})$ time.

\item $B$ plays a symmetric role to $A$; for $0 \leq i' \leq i \leq n+1$ and $0 \leq j \leq n+1$, $B(i',i,j)$ stores the man whom $w_{j}$ most prefers among $\{m_{i'},\dots,m_{i}, \lambda \}$.
$B$ can also be constructed in $O(n^{3})$ time.

\end{itemize}

It is easy to see that $(m_{i'},w_{j'})$ and $(m_{i},w_{j})$ are conflicting if and only if one of the following conditions hold.
Condition 1 can be clearly checked in constant time.
Thanks to the preprocessing, Conditions 2--4 can also be checked in constant time.

\begin{enumerate}

\item $(m_{i'}, w_{j})$ or $(m_{i}, w_{j'})$ is a blocking pair for the matching $\{ (m_{i'}, w_{j'}), (m_{i}, w_{j}) \}$. 

\item $S(i'+1,i-1, j'+1, j-1) =1$. 
(If this holds, there is a blocking pair $(m,w)$ such that $m \in \{ m_{i'+1}, m_{i'+2}, \ldots, m_{i-1} \}$ and $w \in \{ w_{j'+1}, w_{j'+2}, \ldots, w_{j-1} \}$.)

\item $m_{i}$ prefers $A(i,j'+1,j-1)$ to $w_{j}$ or $m_{i'}$ prefers $A(i',j'+1,j-1)$ to $w_{j'}$.
(If this holds, there exists a blocking pair $(m,w)$ such that $m \in \{m_{i'},m_{i}\}$ and $w \in \{w_{j'+1}, \dots, w_{j-1}\}$.)

\item $w_{j}$ prefers $B(i'+1,i-1,j)$ to $m_{i}$ or $w_{j'}$ prefers $B(i'+1,i-1,j')$ to $m_{i'}$.
(If this holds, there exists a blocking pair $(m,w)$ such that $m \in \{m_{i'+1},\dots,m_{i-1}\}$ and $w \in \{w_{j'}, w_{j}\}$). 

\end{enumerate}

This completes the explanation on preprocessing, and from the discussion so far, we have the following theorem:

\begin{theorem}\label{thm:MAX-WSNM}
There exists an $O(n^{4})$-time algorithm to find a maximum cardinality WSNM, given an SMI-instance.
\end{theorem}

\subsection{SMTI}\label{sec:WSNM-SMTI}

Similarly to the SMI case, a weak-WSNM exists in any SMTI-instance: 
Given an SMTI-instance $I$, break all the ties arbitrarily and obtain an SMI-instance $I'$.
Let $M$ be a WSNM of $I'$.
Then it is not hard to see that $M$ is also a weak-WSNM of $I$.
In contrast, there is a simple instance that admits neither a strong- nor a super-WSNM (Fig.~\ref{fig:no-strong-super}).
The empty matching is blocked by any acceptable pair.
The matching $\{ (m_{1}, w_{1}) \}$ is blocked by $(m_{2}, w_{2})$.
The matching $\{ (m_{2}, w_{2}) \}$ is blocked by $(m_{1}, w_{1})$.
The matching $\{ (m_{1}, w_{1}), (m_{2}, w_{2}) \}$ is blocked by $(m_{2}, w_{1})$.
The matching $\{ (m_{2}, w_{1}) \}$ is blocked by $(m_{1}, w_{1})$.

\begin{figure}[ht]
\begin{center}
\renewcommand\arraystretch{1.2}
\begin{tabular}{lllllllllllllllllllllll}
$m_{1}$:  & $w_{1}$ &  & \hspace{15mm} & $w_{1}$:  &  ($m_{1}$ & $m_{2}$) \\
$m_{2}$:  & $w_{1}$ & $w_{2}$ & & $w_{2}$:  &  $m_{2}$
\end{tabular}
\caption{An instance that admits neither a strong-WSNM nor a super-WSNM.}\label{fig:no-strong-super}
\end{center}
\end{figure}

Nevertheless, the algorithm in Sect.~\ref{sec:WSNM-SMI} can be applied to SMTI straightforwardly.
Necessary modifications are summarized as follows:

\begin{itemize}

\item As mentioned above, existence of a WSNM is not guaranteed.
If our algorithm outputs $Y(n+1,n+1)=-\infty$, then it means that no solution exists.

\item The definition of two edges $(m_{i},w_{j})$ and $(m_{x}, w_{y})$ being conflicting must be modified depending on one of the three stability notions.

\item The definitions of the tables $A$ and $B$ need to be modified as follows.
$A(i,j',j)$ stores {\em one of} the women whom $m_{i}$ most prefers among $\{w_{j'},\dots,w_{j}, \lambda \}$.
Similarly, $B(i',i,j)$ stores {\em one of} the men whom $w_{j}$ most prefers among $\{m_{i'},\dots,m_{i}, \lambda \}$.

\item In the SMI case, $A(i,j',j)$ is computed as {\it the better} of $A(i,j',j-1)$ and $A(i,j,j)$ in $m_{i}$'s list.  But now it can happen that $A(i,j',j-1) =_{m_{i}} A(i,j,j)$, in which case $A(i,j',j)$ can be either $A(i,j',j-1)$ or $A(i,j,j)$.  (Strictly speaking, this treatment was needed already in the SMI case because there can be a case that $A(i,j',j-1) = A(i,j,j) = \lambda$, but there we took simplicity.)

\item Conditions 3 and 4 in checking confliction of two edges need to be modified as follows.  In the super- and strong stabilities, ``prefers'' should be replaced by ``weakly prefers''.  In the weak stability, ``prefers'' should be replaced by ``strictly prefers''.  

\end{itemize}

With these modifications, whether two edges are conflicting or not can be checked in constant time. 
Therefore, we have the following corollary:

\begin{corollary}\label{coro:MAX-WSNM-super-strong-weak}
There exists an $O(n^{4})$-time algorithm to find a maximum cardinality super-WSNM (strong-WSNM, weak-WSNM) or report that none exists, given an SMTI-instance.
\end{corollary}

\section{Conclusion}\label{sec:conclusion}

In this paper, we have shown algorithms and complexity results for the problems of determining existence of an SSNM and finding a maximum cardinality WSNM, in the settings both with and without ties.

One of interesting future works is to consider optimization problems.
For example, in SMI we have shown that it is easy to determine if there exists an SSNM with zero-crossing.  
What is the complexity of the problem of finding an SSNM with the minimum number of crossings, and if it is NP-hard, is there a good approximation algorithm for it?
Also, it might be interesting to consider noncrossing stable matchings for other placements of agents, e.g., on a circle or on general position in 2-dimensional plane.




%
%


\begin{thebibliography}{}
%
%
\bibitem{apr05}
Abdulkadiro\v{g}lu, A., Pathak, P.A., Roth, A.E.: The New York City high school match. Am. Econ. Rev. {\bf 95}(2), 364--367 (2005)

\bibitem{aprs05}
Abdulkadiro\v{g}lu, A., Pathak, P.A., Roth, A.E., S\"onmez, T.: The Boston public school match. Am. Econ. Rev. {\bf 95}(2), 368--371 (2005)

\bibitem{atal85}
Atallah, M.J.: A matching problem in the plane. J. of Computer and System Sciences {\bf 31}(1), 63--70 (1985)

\bibitem{bmm12}
Bir\'o, P., Manlove, D.F., McDermid, E.: "Almost stable" matchings in the Roommates problem with bounded preference lists. Theor. Comput. Sci. {\bf 432}, 10--20 (2012)

\bibitem{bmm10}
Bir\'o, P., Manlove, D.F., Mittal, S.: Size versus stability in the marriage problem. Theor. Comput. Sci. {\bf 411}(16-18), 1828--1841 (2010)

\bibitem{clw15}
Chen, D.Z., Liu, X., Wang, H.: Computing maximum non-crossing matching in convex bipartite graphs. Discrete Appl. Math. {\bf 187}, 50--60 (2015)

\bibitem{cook71}
Cook, S.A.: The complexity of theorem-proving procedures. In: Proceedings. STOC 1971, pp. 151--158 (1971)  


\bibitem{gs62}
Gale, D., Shapley, L.S.: College admissions and the stability of marriage. Am. Math. Mon. {\bf 69}(1), 9--15 (1962)

\bibitem{gs85} 
Gale, D., Sotomayor, M.: Some remarks on the stable matching problem. Discrete Appl. Math. {\bf 11}(3), 223--232 (1985)

\bibitem{gi89}
Gusfield, D., Irving, R.W.: The Stable Marriage Problem: Structure and Algorithms. MIT Press, Boston (1989)

\bibitem{hmo20}
Hamada, K., Miyazaki, S., Okamoto, K.: Strongly stable and maximum weakly stable noncrossing matchings. In: Proceedings IWOCA 2020. LNCS, vol.~12126, 304--315 (2020)

\bibitem{ir94}
Irving, R.W.:Stable marriage and indifference. Discrete Appl. Math. {\bf 48}, 261--272 (1994)

\bibitem{imm09}
Irving, R.W., Manlove, D.F., O'Malley, G.: Stable marriage with ties and bounded length preference lists. J. Discrete Algorithms {\bf 7}(2), 213--219 (2009)

\bibitem{ims00}
Irving, R. W., Manlove, D. F., Scott, S.: The hospitals/residents problem with ties. In: Proceedings SWAT 2000. LNCS, vol.~1851, pp. 259--271 (2000)  

\bibitem{ims03}
Irving, R.W., Manlove, D.F., Scott, S.: Strong stability in the hospitals/residents problem. In: Proceedings STACS 2003. LNCS, vol.~2607, pp. 439-450 (2003) 

\bibitem{kt86}
Kajitani, Y., Takahashi, T.: The noncross matching and applications to the 3-side switch box routing in VLSI layout design. In: Proceedings IEEE International Symposium on Circuits and Systems, pp.~776--779 (1986)

\bibitem{kmmp04}
Kavitha, T., Mehlhorn, K., Michail, D., Paluch, K.: Strongly stable matchings in time $O(nm)$ and extension to the hospitals-residents problem. ACM Trans. Algorithms {\bf 3}(2) (2007). Article No.~15

\bibitem{knuth76}
Knuth, D.E.: Mariages Stables, Les Presses de l'Universit\'e Montr\'eal, (1976).  (Translated and corrected edition, Stable Marriage and Its Relation to Other Combinatorial Problems, CRM Proceedings and Lecture Notes, Vol. 10, American mathematical Society, 1997.)

\bibitem{mop93}
Malucelli, F., Ottmann, T., Pretolani, D.: Efficient labelling algorithms for the maximum noncrossing matching problem. Discrete Appl. Math. {\bf 47}(2), 175--179 (1993)

\bibitem{man13}
Manlove, D.F.: Algorithmics of Matching under Preferences. World Scientific, Singapore (2013)

\bibitem{mo19}
Miyazaki, S., Okamoto, K.: Jointly stable matchings. J. Comb. Optim. {\bf 38}(2), 646--665 (2019)

\bibitem{roth84}
Roth, A.E.: The evolution of the labor market for medical interns and residents: a case study in game theory. J. Political Economy {\bf 92}(6), 991--1016 (1984)

\bibitem{roth86}
Roth, A.E.: On the allocation of residents to rural hospitals: a general property of two-sided matching markets. Econometrica {\bf 54}(2), 425--427 (1986)

\bibitem{rs90}
Roth, A.E., Sotomayor, M.: Two-sided Matching: A Study in Game-theoretic Modeling and Analysis. Cambridge University Press, Cambridge, (1990)

\bibitem{ri19}
Ruangwises, S., Itoh, T.: Stable noncrossing matchings. In: Proceedings IWOCA 2019. LNCS,  vol.~11638, pp.~405--416 (2019)

\bibitem{tov84}
Tovey, C.A.: A simplified NP-complete satisfiability problem. Discrete Appl. Math. {\bf 8}, 85--89 (1984)

\bibitem{ww85}
Widmayer P., Wong, C.K.: An optimal algorithm for the maximum alignment of terminals. Information Processing Letters {\bf 20}(2), 75--82 (1985)

\end{thebibliography}


\end{document}